\documentclass[11pt, oneside]{article}
\usepackage[margin=1in]{geometry}

\usepackage{siunitx}
\usepackage{subfig}
\usepackage{graphicx}
\usepackage{url}
\usepackage{amsmath,amsthm,amsfonts,mathtools}
\usepackage{bbm}
\usepackage{setspace}
\usepackage{subfloat}
\usepackage[round]{natbib}

\newcommand{\indep}{\perp\!\!\!\perp} 
\newcommand{\argmin}{\operatornamewithlimits{argmin}} 
\def\inprob{\stackrel{p}{\rightarrow}} 
\def\indist{\rightsquigarrow} 

\newtheorem{theorem}{Theorem}[section]
\newtheorem{lemma}{Lemma}[section]

\author{Maria Cuellar\\
Department of Criminology\\
University of Pennsylvania\\
Philadelphia, PA 19104\\
mcuellar@sas.upenn.edu
\and
Edward H. Kennedy \\
Department of Statistics \& Data Science\\
Carnegie Mellon University\\
Pittsburgh, PA 15213\\
edward@stat.cmu.edu
}

\title{A nonparametric projection-based estimator for the probability of causation, with application to water sanitation in Kenya}

\begin{document}
\maketitle



\begin{abstract}
Current estimation methods for the probability of causation (PC) make strong parametric assumptions or are inefficient. We derive a nonparametric influence-function-based estimator for a projection of PC, which allows for simple interpretation and valid inference by making weak structural assumptions. We apply our estimator to real data from an experiment in Kenya. This experiment found, by estimating the average treatment effect, that protecting water springs reduces childhood disease. However, before scaling up this intervention, it is important to determine whether it was the exposure, and not something else, that caused the outcome. Indeed, we find that some children, who were exposed to a high concentration of bacteria in drinking water and had a diarrheal disease, would likely have contracted the disease absent the exposure since the estimated PC for an average child in this study is 0.12 with a 95\% confidence interval of (0.11, 0.13). Our nonparametric method offers researchers a way to estimate PC, which is essential if one wishes to determine not only the average treatment effect, but also whether an exposure likely caused the observed outcome.


\end{abstract}


\section{Introduction}

The probability of causation (PC) is the probability that an outcome was caused by a specific exposure, and not by something else. Researchers have suggested estimating PC to answer questions of causality in which the exposure and outcome have already been observed for a group of individuals. Current methods for estimating PC and its confidence intervals use plugin estimators, which are inefficient, and parametric assumptions, which are often difficult to make correctly. In this article, we present a novel estimator that allows for nonparametric estimation and derivation of valid confidence intervals under weak structural assumptions. We illustrate our method in an application to determine whether, for children in Western Kenya who were exposed to high concentrations of bacteria in their drinking water, it was the bacteria or something else that caused their diarrheal disease. 

PC, sometimes called the probability of necessity, has been of interest for some time (\cite{mosteller, tian, pearl, pearl2014, dawidfaigmanfienberg2013, dawidmusiofienberg2016}) in the law and in epidemiology because it is especially useful in questions of ``but for'' causation, that is, in determining whether a specific outcome would not have occurred in the absence of that exposure. In other words, it is the probability that a certain outcome can be attributed to a certain exposure.

PC is especially useful whenever there is a harmful exposure and a negative outcome. For instance, suppose a man has been exposed to a harmful chemical at his work, and then he develops cancer. How can an expert witnesses testifying in court determine whether his cancer was caused by the chemical exposure? PC represents the probability that it was the exposure, and not something else, that caused the cancer. 
Formally, it is defined as
\begin{equation} 
\label{eq:PC}
PC(x) = P(Y^0=0 \mid Y=1, A=1, X=x),
\end{equation}
where $Y^0$ is the potential outcome when the exposure is set to zero, $Y$ is the observed outcome, $A$ is the exposure, and $X$ are the observed covariates. In the cancer example, this is the probability that the man would have failed to develop the cancer had he not been exposed to the chemical, given that he had cancer when he was exposed. Note that this is not equivalent to the average treatment effect on the treated because it conditions on the treated \emph{and} on those for whom a specific outcome has been observed. In a lawsuit setting, \cite{mosteller} proposed that PC should correspond to the percentage out of the full compensation that would be awarded to the person bringing the suit. \cite{dawidfaigmanfienberg2013} and \cite{dawidmusiofienberg2016} suggest instead that PC should be interpreted as my degree of belief about the attribution, a Bayesian interpretation.

PC can also be used in public health, as we do in our application. Children in Western Kenya were exposed to a high bacterial concentration in their drinking water, and then developed a diarrheal disease. What is the probability that it was the bacterial exposure, and not something else, that caused the disease? PC answers this question. Knowing whether the exposure causes the outcome on average, which is the usual question asked in randomized trials, is important. Knowing whether the exposure \emph{actually} caused the outcome (in those children who were exposed and had the outcome) is also important, especially if a policy is to have a desired impact on a target population.

In this article we seek to answer whether, for those children who were exposed to a high concentration of bacteria in their drinking water and became ill, the exposure is what caused their illness, and not something else. But estimating PC is not trivial. First, the reader might have noticed that $PC(x)$ requires knowing both the outcome under exposure and under no exposure, only one of which could be observed directly. Thus identification assumptions are required to be able to estimate PC. Then, an estimator must be used to obtain estimates and a measure of uncertainty for these estimates, such as confidence intervals. With an estimation problem that could have serious consequences on an individual's future, as it would if it is used in a legal trial or to make health policy decisions, it is essential to estimate PC in a way that gets as close to the truth as possible, and has a valid measure of uncertainty.

The current methods used to find estimates, parametric and nonparametric plugin estimators, are useful but have some drawbacks. Parametric plugin estimators---such as linear and logistic regression---allow for valid inference, but only if the parametric assumptions of the model are correct. On the other hand, simple nonparametric plugin estimators---for example based on random forests or support vector machines---allow the researcher to fit models without making parametric assumptions, but at the cost of slower convergence rates and not allowing for the construction of confidence intervals (\cite{tsiatis}).

The main methodological contribution of this article is the derivation of an estimator for PC that avoids untestable assumptions and has a valid measure of uncertainty. We derive a nonparametric influence-function-based (IFB) estimator, which targets a projection of the true PC function onto a parametric model. Our IFB estimator does not require any parametric assumptions for the nuisance functions and it can be used to derive valid confidence intervals because it is asymptotically normal, as long as the nuisance functions can be estimated at a relatively slow rate of convergence of order $o_\mathbbm{P}(n^{-1/4})$. In addition, the estimator does not require knowledge of propensity scores, and it avoids the limitations of nonparametric plugin estimators, which yield slow convergence rates compared to IFB estimators. Although obtaining valid confidence intervals without making parametric assumptions might sound somewhat magical, note that the ``catch'' is that our estimator relies on the nonparametric methods having a certain rate of convergence, and that we target a projection of the true function onto a parametric model. 

The method derived in this article also provides a novel way of analyzing data from randomized trials (and other datasets following the assumptions we listed), and thus it could be used to find new insights about causation and public policy.

\begin{figure}[ht]
\centering
\includegraphics[width=.9\textwidth]{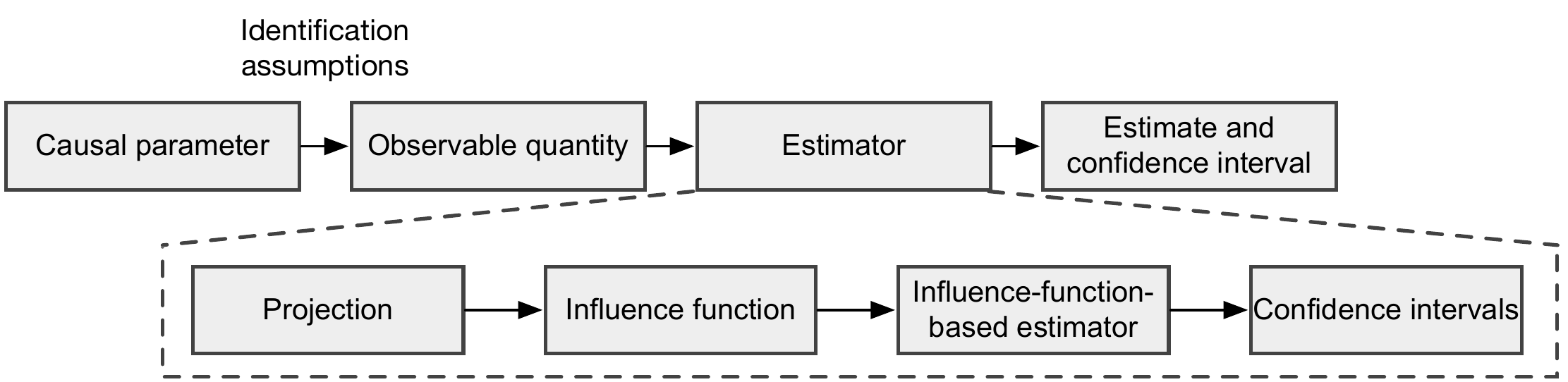}
\caption{Procedure we use to obtain estimates of the probability of causation, the causal parameter shown in Equation 1.} 
\label{fig:estimationprocess}
\end{figure}

The procedure we follow to derive estimates and confidence intervals for PC is shown in Figure \ref{fig:estimationprocess}. We define the causal parameter, make identification assumptions to obtain an observable quantity, and derive the estimator. For the estimator, we define a projection onto a parametric model, derive its influence function, derive the influence-function-based estimator, and derive its confidence intervals. Finally, we can use the estimator to obtain an estimate and confidence interval for PC. This final estimate will depend on which parametric model is selected on which to project. For instance, we use logistic regression in our application in this article, so the final estimates and confidence intervals will have estimates of PC and estimates for the coefficients of each covariate.

The remainder of this article is organized as follows. Section \ref{sec:backgroundapp} provides the background for our application. Section \ref{sec:estimator} shows the derivation of the estimator for the probability of causation. This includes the identification procedure, the projection approach, the proposed influence-function-based estimator, and the estimator's asymptotic properties. The proof that the influence-function-based estimator is asymptotically normal is left in the appendix. Section \ref{sec:simulation} compares the performance of the influence-function-based estimator to the performance of plugin estimators by simulation. Section \ref{sec:application} provides the application with real data from a randomized controlled trial in Kenya. Section \ref{sec:discussion} offers a discussion of our method and our results.


\section{Background}\label{sec:backgroundapp}

Diarrheal diseases are a leading cause of disease and mortality in the developing world, and for children under 5, these diseases account for 20 percent of deaths (\cite{jpal}). These diseases are often transmitted when a water supply is contaminated with fecal matter. In rural Kenya, 43 percent of the population gets their drinking water from nearby springs (as shown in Figure \ref{fig:jpalpics}, left), usually after transporting it in 10- to 20-liter jerry-cans. In springs, where water seeps out of the ground, the water is vulnerable to contamination when people dip their cans to scoop out water and when runoff introduces human or animal waste into the area.

\begin{figure}[ht]
\centering
\subfloat{\includegraphics[width=.35\textwidth]{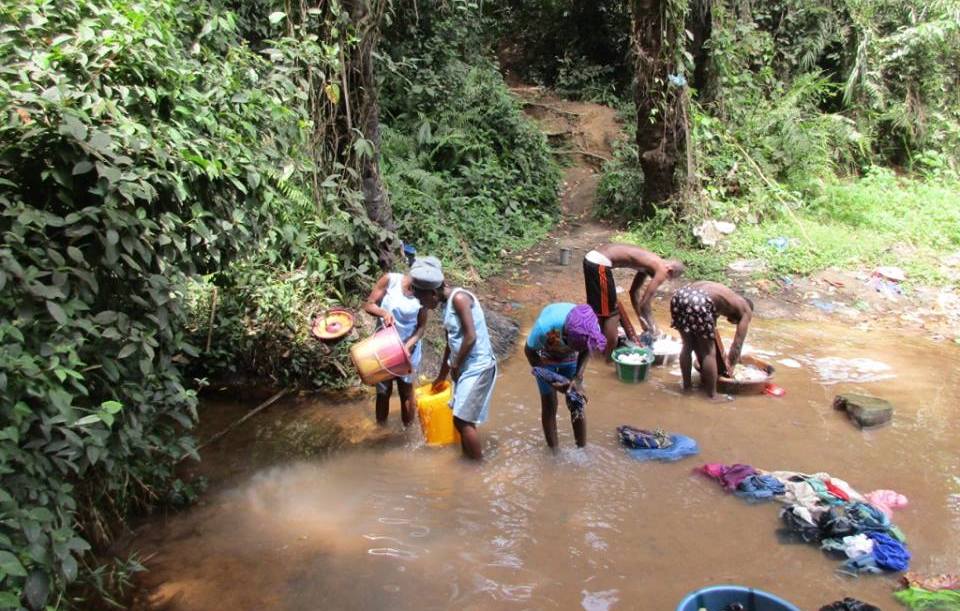}} \qquad
\subfloat{\includegraphics[width=.35\textwidth]{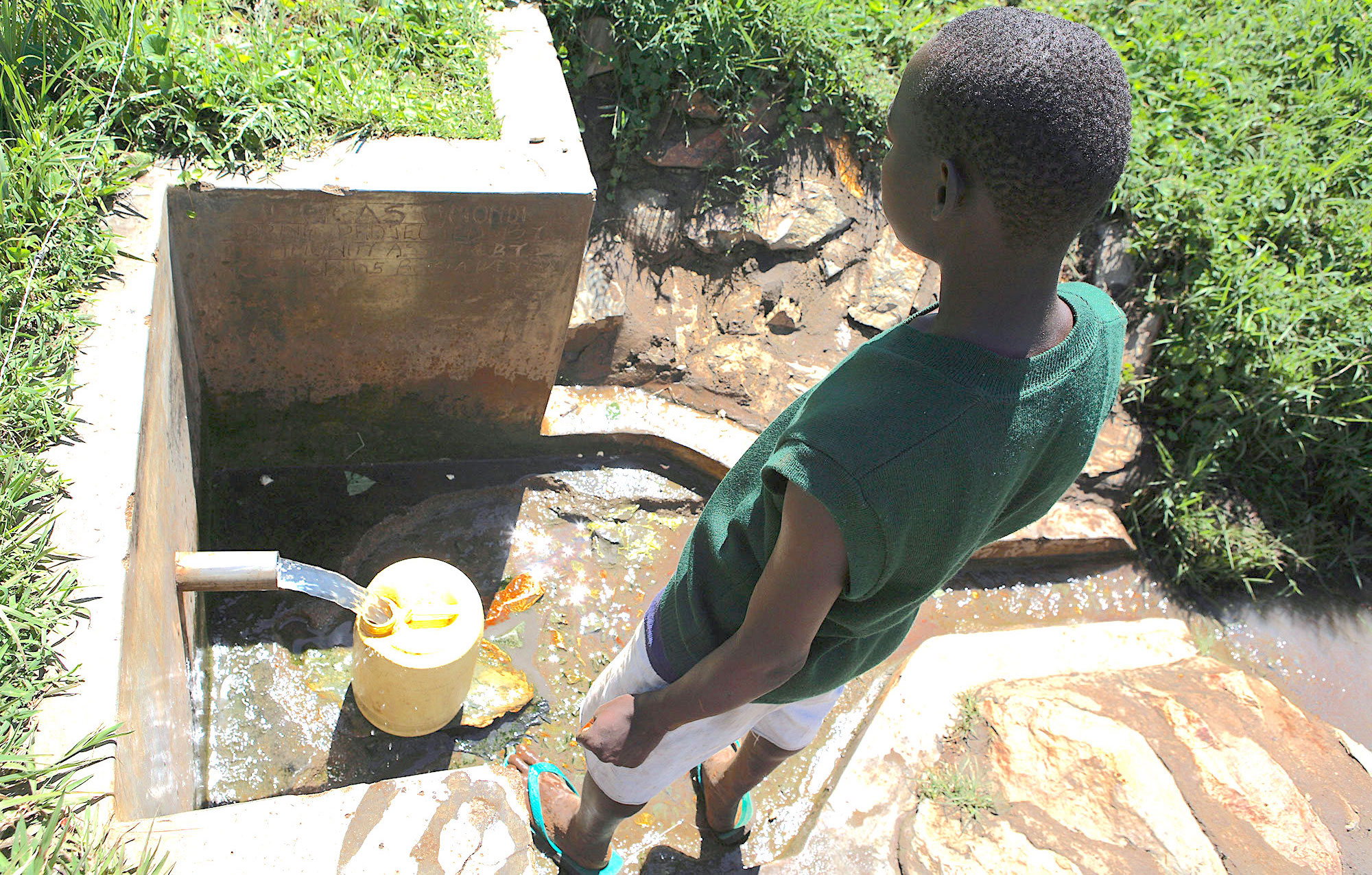}}
\caption{Images from the Busia district of Western Kenya, where 43\% of individuals access water for drinking and cooking from unprotected springs (left). International Child Support, a local NGO, builds concrete structures (right) for residents to access cleaner water. From https://www.povertyactionlab.org/evaluation/cleaning-springs-kenya.}
\label{fig:jpalpics}
\end{figure} 

To prevent the outcome of diarrheal diseases in children in the Busia district of Western Kenya, a local NGO called International Child Support built protective cement structures around a randomly selected group of springs (as shown in Figure \ref{fig:jpalpics}, right). The structures forced water to flow through a pipe rather than seeping from the ground and thus helping prevent the transfer of bacteria. \cite{jpal}, researchers from Abdul Latif Jameel Poverty Action Lab (JPAL), worked with the NGO to deploy the experiment as a randomized controlled trial, and thus JPAL was able to estimat the causal effect of building the cement structures on the bacterial concentration in the water and in diarrheal diseases in the community. 

In the experiment, 400 springs were randomized, which resulted in 13,036 households drinking from protected springs and 32,514 from unprotected springs. Household characteristics such as income, education and health, were approximately equal among the two groups at the start of the program, suggesting that there were no systematic differences between communities that had their springs protected and those that did not. During the trial, the researchers collected measures on the level of water contamination and diarrheal disease in all communities. Water quality was measured at all sample springs and households using protocols based on those used at the United States Environmental Protection Agency (EPA). The water quality measure we use is contamination with E. coli, an indicator bacterium that is correlated with the presence of fecal matter, as measured by the natural log of the most probable number (MPN) of colony-forming bacteria per 100 ml of water. 

The JPAL researchers sought to estimate the average treatment effect. They found that spring protection significantly reduces diarrhea for children under age three (at baseline or born since the baseline survey) by 25\%. They also found that spring protection reduces fecal contamination at the spring by a much higher rate of 66\%. However, some of these beneficial effects on the child's health are lost because household water quality improves less than the quality at the source, due to re-contamination (\cite{jpal}). Interestingly, diarrhea reduction was disproportionately concentrated among girls, suggesting that cleaning water springs could be an effective tool for the improvement of female child survival.

For the purposes of our project, we take the JPAL direct outcome of bacterial concentration as our exposure $A$, and their indirect outcome of diarrheal diseases as our outcome $Y$. We require having a binary exposure, so for us, the exposure is to a higher amount of bacteria due to having no cement structure (denoted by $A=1$), and the control is exposure to a low amount of bacteria due to having a cement structure (denoted by $A=0$). It is left as future research to derive this estimator for continuous exposures. Defining the exposure and outcome in this way allows us to see the experiment as a situation in which there was a harmful exposure (to high concentrations of bacteria) and a negative outcome (diarrheal disease). 

The question that must be answered in order to allocate resources to address the problem effectively for specific sub-populations is, ``How likely is it that the negative outcome was caused by the exposure, and not by something else?'' To reach this estimate, we need an estimator for PC, which we introduce in the next section.


\section{Deriving an estimator for the probability of causation}\label{sec:estimator}

Throughout this paper we assume access to an iid sample $(Z_1,...,Z_n) \sim P$, where $Z=(X,A,Y)$ for $X \in \mathbbm{R}^d$ a vector of covariates, $A \in \{0,1\}$ an indicator of binary exposure, and $Y \in \{0,1\}$ an outcome indicator. We also let $Y^a$ denote the potential outcome that would have been observed under exposure level $A=a$. For notational simplicity we let 
\begin{equation} 
\mu_a(x) = \mathbbm{E}(Y \mid X=x, A=a) \ , \ \pi(x) = P(A=1 \mid X=x) 
\end{equation}
denote the outcome regression and propensity score functions, respectively.

\subsection{Identification}

For identification, we make four causal assumptions, which are similar to the ones made by \cite{dawidmusiofienberg2016} and \cite{tian}: 

\begin{enumerate}
  \item \textbf{Positivity}: $P(A=a \mid X=x) > 0$ for all $x$ with $P(X=x) \neq 0$.
  
  Positivity---the assumption that all individuals have a chance of being exposed---is generally necessary for estimating causal effects. It implies that there cannot be some individuals in the sample that cannot be exposed. 
  
  \item \textbf{Consistency}: $A=a \implies Y^a=Y$.
  
  Consistency helps us translate from potential outcomes to observed outcomes. It can be violated, for example, if there is interference. Since an individual's outcome might be influenced by someone else's exposure, that individual's outcome might not be equal to her potential outcome under that exposure. 

  \item \textbf{No unobserved confounders}: $Y^0 \indep A \mid X$.
  
  No unobserved confounders---sometimes called exchangeability---is the assumption that the treated and the untreated individuals are exchangeable because the treated, had they remained untreated, would have experienced the same average outcome as the untreated did, and vice versa. It holds by design in a randomized trial. Note that here we only require it for the potential outcome $Y^0$, which is a weaker assumption than the usual joint independence of $(Y^0,Y^1) \indep A|X$. See \cite{tian} for identification of the probability of causation (although the authors call it the probability of necessity) if one relaxes this assumption.

  \item \textbf{Monotonicity}: $Y^1 \geq Y^0$.
  
  Monotonicity implies that the outcome under exposure is always higher than the outcome under no exposure. For harmful exposures we can expect monotonicity to hold: We do not expect that exposing an individual to a high dose of radioactivity would remove an individual's radiation sickness; it could only make it worse. For some applications, monotonicity can be a strong assumption. Here, it allows us to obtain a point estimate for PC. However, if the user chooses not to assume monotonicity, the identified parameter shown below in \eqref{eq:idparameter} is simply the lower bound of PC. This is discussed in \cite{dawidfaigmanfienberg2013} and \cite{dawidmusiofienberg2016}.

\end{enumerate}
Whether the assumptions apply to a specific analysis depends on the data. Using these assumptions, PC equals one minus the risk ratio, a quantity we call $\gamma(x)$,
\begin{align}
PC(x) & = P(Y^0 =0 | Y=1, A=1, X=x) \nonumber \\
& = \frac{P(Y^0=0, Y=1 \mid A=1, X=x) }{P(Y=1 \mid A=1, X=x)} \nonumber\\
&= \frac{P(Y - Y^0 = 1 \mid A=1, X=x)}{P(Y=1 \mid A=1, X=x)} \nonumber\\
& = 1 - \frac{\mathbb{E} (Y \mid A=0, X=x)}{\mathbb{E}(Y \mid A=1, X=x)} \nonumber\\
&= 1 - RR(x) \nonumber\\
& \equiv \gamma(x). \label{eq:idparameter}
\end{align}

Where the risk ratio is $RR(x) = \frac{\mathbb{E} (Y \mid A=0, X=x)}{\mathbb{E}(Y \mid A=1, X=x)} = \frac{\mu_0 (x) }{\mu_1 (x)}$. Sometimes the risk ratio is written as $\mu_1(x)/\mu_0(x)$, but we define it as its inverse for convenience. Note that the new method we propose for estimating $\gamma(x)$ can also be used to estimate the risk ratio. Also note that $\gamma(x)$ does not include potential outcomes, and thus it can be estimated with observed data. 

For a continuous treatment, the identification procedure would follow the same logic, but there PC would be a curve in $A$ given by
\begin{equation}
P(Y_0=0 | Y=1, A=a, X=x)=1 - \frac{\mathbb{E}(Y | X, A=0)}{\mathbb{E} (Y | X, A=a)}.
\end{equation}
Estimation could proceed as in \ref{sec:ifs} but the working model would also be indexed by A, and the influence function would change slightly and depend on the conditional density of $A$ given $X$.

\subsection{Motivation for using an influence-function-based projection estimator}

The simplest (and most common) method for estimating RR is with a plugin estimator (see, for example, \cite{epi-stephen}), which works by estimating the nuisance functions via generic regression methods and then ``plugging in'' the estimated regression function into the functional. Thus the plugin estimator of RR is
\begin{equation}
\label{eq:RRplugin}
\widehat{RR}_{\text{PI}}(x) =\frac{\hat{\mu}_0 (x) }{\hat{\mu}_1 (x)}.
\end{equation}

Although the nuisance functions $\mu_0(x), \mu_1(x)$ could be estimated parametrically (for example, assuming a logistic parametric relationship between the outcome and the treatment and covariates), if the parametric assumptions made by the user are incorrect, the parametric plugin estimator will yield biased estimates. This is especially problematic for the probability of causation, which could be used for sensitive applications such as legal trials (\cite{mosteller}). 

Thus, we do \emph{not} want to assume that the true $\gamma(x)$ function follows a known parametric form, i.e. that $\gamma(x)=g(x;\beta)$ is completely known up to some finite-dimensional parameter $\beta \in \mathbbm{R}^d$. Others, such as \cite{linbowang} and \cite{tchetgen2013}, tend to assume that the true $\gamma(x)$ follows a parametric model; however, a disadvantage is that the assumed parametric model will very likely be incorrect. To estimate RR \cite{linbowang} propose the conditional log odds-product as a preferred nuisance model. Their approach is to develop an unconstrained nuisance model that is variation independent of RR (or RD, the risk difference). They propose doubly-robust estimators of models for (monotone transformations of) RR and RD that are consistent and asymptotically normal even when the nuisance model is misspecified, provided that they have correctly specified the model for the propensity score, $\pi = \mathbbm{E}(A \mid X)$. Doubly-robust estimators combine a form of outcome regression with a model for the exposure (i.e., the propensity score) to estimate the causal effect of an exposure on an outcome. However, the approach in \cite{linbowang} still relies on parametric assumptions, both for the nuisance functions and for the RR itself.

Using a nonparametric plugin estimator (e.g., using random forests) does not allow for the estimation of valid confidence intervals in general without making stronger assumptions. Thus, we use influence functions, which will allow us to derive valid confidence intervals without making parametric assumptions and by making weak structural assumptions.

In order to estimate $\hat{\gamma}(x)$ nonparametrically and still obtain valid confidence intervals under weak structural assumptions, we derive an influence-function-based estimator for $\gamma(x)=1-RR(x)$. However, in a nonparametric model, $\gamma(x)$ does not have an influence function because it is not pathwise differentiable (\cite{bickel}); thus it is similar to a density or regression function, which cannot be estimated at $\sqrt{n}$ rates in a nonparametric model.

We use a weighted least squares projection of the true function $\gamma(x)$ onto a parametric model $g(x;\beta)$. By doing this, we can derive an influence-function-based estimator for the projection, which will allow us to obtain valid confidence intervals for the projection with nonparametric estimation of the nuisance functions, under weak structural assumptions. To clarify, we target a projection of $\gamma(x)$ onto $g(x;\beta)$ (rather than assuming $\gamma(x)=g(x;\beta)$) because we want to be honest about the fact that we do not necessarily have a correct parametric model for $\gamma(x)$. In this article, we are using a parametric model $g(x;\beta)$ simply as a tool that is useful for summarizing the data, which is generated according to some true possibly complex form for $\gamma(x)$. A reader who is more comfortable with assuming that $\gamma(x)=g(x;\beta)$ might think it would be easier to interpret what the estimates of PC and the coefficients meant if one assumed the data followed a parametric model. But what if this model is incorrect? Then assuming the data follow a parametric model yields estimates that are very difficult to interpret, whereas the projection will yield estimates that are correct. 

Our proposed estimator, as we will show in detail below, is the value $\hat{\beta}$ that satisfies $\hat{\Psi}(\hat{\beta})=0$, where $\hat{\Psi}(\hat{\beta}) = \mathbbm{P}_n (\varphi (Y,A,X))$, $\mathbbm{P}_n$ is a sample average and $\varphi(Y,A,X)$ is the influence function of the projection, which is a function of the nuisance parameters. This estimator is implicitly defined because it is the value of $\hat{\beta}$ that sets $\hat{\Psi}$ equal to zero. Thanks to the properties of influence functions, the estimator will yield valid confidence intervals under weak structural assumptions despite estimating all nuisance functions with nonparametric methods.

Influence functions were originally introduced articles about robustness (\cite{hampel, huber}). To first order, $\varphi$ is the influence of the $i$-th observation on an estimator $\hat\psi$. Their importance in diverse semi- and nonparametric functional estimation problems was established starting in the 1990s (\cite{bickel, van2003unified, vandervaart, tsiatis, robins2008}). Influence functions have also played a fundamental role in causal inference recently, as well as in economics (\cite{chernozhukov}) and machine learning (\cite{kandasamy}). See \citet{robins2000, ogburn, toth, kennedy2016} for just a few interesting examples.

\subsection{Projection}

We target a projection of the true $\gamma(x)$ function onto a parametric model that does not make any assumptions about the data distribution. In other words we redefine our parameter of interest as the best-fitting approximation of the true function $\gamma(x)$ by a parametric model $g(x;\beta)$, where $\beta$ is a vector of real parameters. For example one could use a best fitting linear approximation with $g(x;\beta) = \beta^T x$, based on least squares error; more details and examples are given shortly. Then, we estimate the parameters of the projection with an influence-function-based (IFB) approach. The IFB approach requires estimation of nuisance functions, which can be accomplished with parametric or nonparametric methods. We propose using nonparametric methods, the more flexible option; it turns out we are nonetheless still able to attain fast parametric rates for estimating $\beta$.

We define the projection as follows. Let
\begin{equation}
\label{eq:projection}
\beta = \argmin_{\beta^*} \mathbbm{E} \left[ w(X) (\gamma(X) - g(X;\beta^*))^2) \right].
\end{equation}
Note $\beta$ is a function of the distribution $P$, the true distribution from which our samples are drawn, since it minimizes the $L_2$ error between the model and the true $\gamma(x)$ function. We use an $L_2$ loss-based projection because of its simplicity, but other loss functions could be used instead. The function $w(X)$ is a user-specified weight function, and can be defined based on subject matter concerns; for example, if certain parts of the covariate space are more important to approximate well, this can be incorporated into the weight. If all are equally important, one could simply choose a uniform weight $w(x)=1$. In the case that the model $g(x;\beta)$ is correct, different choices of the weight give estimators of different efficiencies; otherwise the weight simply defines the projection. One could also consider a data-dependent projection, where the weight equals a kernel function that collapses to a point $X=x$ as sample size grows; this is similar to the approach taken in \cite{kennedy_ctstreatment}.

Under standard regularity conditions (\cite{tsiatis}), the $\beta$ defined in \eqref{eq:projection} satisfies the moment condition
\begin{equation}
\label{eq:projection2}
\mathbbm{E} \left[ \frac{\partial g(X;\beta)}{\partial \beta} w(X) (\gamma(X) - g(X;\beta))) \right] : = \Psi ( \beta;P) = 0.
\end{equation}
We write $\Psi(\beta;P) = \Psi (\beta) = 0$ for short. This moment condition implicitly defines our parameter $\beta$, i.e., the coefficients in our best approximation to $\gamma(x)$ based on the model $g(x;\beta$). For this article, we assume the solution $\beta$ is unique.

We define our estimator $\hat\beta$ as the solution to 
\begin{equation}
\label{eq:estimator}
\hat\Psi(\hat\beta) = 0
\end{equation}
for an estimated version of the moment condition, to be defined shortly. Therefore the problem reduces to estimating the moment condition $\Psi(\beta)$ at a given fixed $\beta$ value. This differs from classical M-estimation problems (\cite{vandervaart}) in that the moment condition is a complex functional, more complicated than a simple expectation since it involves unknown complex nuisance functions $\mu_a,\pi$.

A simple way to estimate $\Psi (\beta)$ would be with a plugin estimator
\begin{equation}
\widehat\Psi (\beta) = \mathbbm{P}_n \left[ \frac{\partial g(X;\beta)}{\partial \beta} w(X) (\widehat\gamma(X) - g(X;\beta))) \right],
\end{equation}
(where $\mathbbm{P}_n$ is a sample average), that replaces the expectation with an empirical sample average, and (more problematically) replaces $\gamma(x)$ with its estimated version $1-\widehat\mu_0(x)/\widehat\mu_1(x)$. But this has the usual problems of a plugin (i.e. slower convergence rates if $\gamma$ is estimated nonparametrically, and unreliable parametric assumptions otherwise). 

In contrast, we derive an influence function-based estimator of $\Psi(\beta)$, which allows for nonparametric estimation of the complex nuisance functions, while still yielding $\sqrt{n}$-rates, asymptotic normality, and valid confidence intervals both for the moment condition itself and for the value $\widehat\beta$ that sets it to zero.

\subsection{Proposed estimator}\label{sec:ifs}

Here we give the influence function for the moment condition parameter $\Psi(\beta^*)$ defined in \eqref{eq:projection}, which has two crucial consequences: first its variance provides an asymptotic efficiency bound, indicating that no regular asymptotically linear estimator can have smaller MSE without adding assumptions, and second it tells us how to construct efficient estimators that attain such a bound under weak nonparametric conditions. 

One can speak of an influence function for an estimator or for a parameter. When we say an estimator has a particular influence function, this means the estimator is asymptotically equivalent to a sample average of the influence function, i.e., 
\begin{equation}
\label{eq:simpleDef}
\hat\psi - \psi = \mathbbm{P}_n ( \varphi(Z) ) + o_p(1/\sqrt{n}),
\end{equation}
where the influence function $\varphi(Z)$ has mean zero and finite variance $\mathbbm{E}(\varphi(Z)\varphi(Z)^T)$, and $\mathbbm{P}_n$ denotes a sample average. Note the above formulation immediately allows for constructing Wald-style confidence intervals based on the central limit theorem. When we say a parameter has a particular influence function, this indicates that that function acts as a pathwise derivative of the parameter, where paths refer to parametric submodels (\cite{vandervaart}). The influence function behaves as the derivative term in a von Mises-style distributional Taylor expansion of the parameter. In practice, influence functions are crucial because their variance gives an efficiency benchmark, and because they can be used to construct estimators that are nonparametrically efficient with nice properties such as double-robustness. In particular a standard way to construct such estimators is to solve an estimating equation based on the (estimated) influence function.

Our first result gives the efficient influence function of the moment condition $\Psi(\beta^*)$, which we then use to construct an estimator $\hat\beta$. A proof of this result is given in the Appendix.

\begin{theorem} \label{thm:eif}
 Under a nonparametric model, the (uncentered) efficient influence function for the moment condition $\Psi(\beta^*)$ at any fixed $\beta^*$ is given by
\small
\begin{equation}
\label{eq:myIF}
\varphi(Z;\beta,\eta) = h({X};\beta^*) \left[ \frac{1}{\mu_1({X})} \left( \frac{\mu_0({X})}{\mu_1({X})} \frac{A(Y-\mu_1({X}))}{\pi({X}) }
- \frac{(1-A)(Y-\mu_0({X}))}{(1-\pi({X}))} \right) + \gamma(X) - g(X;\beta^*) \right] , 
\end{equation}
\normalsize
where we define $\eta=(\pi,\mu_0,\mu_1)$ and $h(x;\beta) = \frac{\partial g(x;\beta)}{\partial \beta} w(x)$.
\end{theorem}

Now we can solve an estimating equation based on the above influence function to construct an appropriate estimator with desirable efficiency and robustness properties. Since the influence function is linear in the parameter $\Psi(\beta^*)$, this is equivalent to computing the sample average of an estimate of \eqref{eq:myIF}. Solving for the $\hat\beta$ that sets the estimated moment condition equal to zero consists of solving the estimating equation.

Therefore our proposed estimator of the minimizer $\beta$ in \eqref{eq:projection} is defined as the assumed unique value $\hat\beta$ that satisfies
\begin{equation}
\hat\Psi(\hat\beta) = 0 
\end{equation}
where  
\begin{align}
\hat\Psi(\beta) = & \mathbbm{P}_n \bigg[ h({X};\beta) \bigg\{ \frac{\hat\mu_0({X})}{\hat\mu_1({X})^2} \frac{A}{\hat\pi(X)} \bigg(Y-\hat\mu_1({X})\bigg)
- \frac{1}{\hat\mu_1({X})} \left( \frac{1-A}{1-\hat\pi({X})}\right) \bigg(Y-\hat\mu_0({X})\bigg) \nonumber \\ 
& + \bigg(\gamma(X) - g(X;\beta) \bigg) \bigg\} \bigg],
\end{align} 
where $\hat\gamma(X)=1-\hat\mu_0(X)/\hat\mu_1(X)$. Of course from this one can also construct estimates $g(x;\hat\beta)$ of predicted values based on the best-fitting approximation. In the next section we analyze the asymptotic properties of our proposed estimator.

\subsection{Asymptotic properties}

Next, we analyze the asymptotic behavior of our influence-function-based (IFB) estimators $\hat\beta$ and $g(x;\hat\beta)$. We show that they can be $\sqrt{n}$-consistent and asymptotically normal even when the nuisance functions $(\pi,\mu_a)$ are estimated data-adaptively with flexible nonparametric regression tools. The proof of the result is similar to the proof of Theorem 5.31 from \citet{vandervaart}, which analyzes general M-estimators in the presence of complex nuisance functions. Throughout, for a possibly data-dependent function $f(z)$ we let $\| f \|^2 = \int f(z)^2 \ dP(z)$ denote the squared $L_2(P)$ norm.

\begin{theorem} \label{thm:largesample}
Assume that:
\begin{enumerate}
\item The sequence of functions $\hat\varphi_n = \varphi(\cdot;\hat\beta,\hat\eta)$ and its limit $\varphi=\varphi(\cdot;\beta,\eta)$ are contained in a Donsker class with $||\hat\varphi_n - \varphi||=o_\mathbbm{P}(1)$.
\item The map $\beta \rightarrow \mathbbm{E}\{ \varphi(Z;\beta,\eta) \}$ is differentiable at the true $\beta$ uniformly in $\eta$, with invertible derivative matrix $M$ (evaluated at the true $\beta$ and $\eta$).
\end{enumerate}
Then as long as $(\hat\beta,\hat\eta) \inprob (\beta,\eta)$ the proposed estimator is consistent with rate of convergence
$$ | \hat\beta - \beta | = O_\mathbbm{P}\left\{ \frac{1}{\sqrt{n}} + \Big( \| \hat\pi - \pi \| + \| \hat\mu_1 - \mu_1 \| \Big) \sum_{a=0}^1 \| \hat\mu_a - \mu_a \| \right\} . $$
Suppose further that:
\begin{enumerate}
\item[3.] $\Big( \| \hat\pi - \pi \| + \| \hat\mu_1 - \mu_1 \| \Big) \sum_a \| \hat\mu_a - \mu_a \|=o_\mathbbm{P}(1/\sqrt{n})$.
\end{enumerate}
Then the proposed estimator attains the nonparametric efficiency bound, and is asymptotically normal with
$$ \sqrt{n} (\hat\beta - \beta) \indist N \Big(0, M^{-1} \mathbbm{E} ( \varphi \varphi^T) (M^T)^{-1} \Big), $$
and similarly for any fixed $x$ we have
\begin{equation}
\label{eq:sandwichvariance}
\sqrt{n} \left( g(x; \hat\beta) - g(x; \beta) \right) \indist 
N \left( 0, \left( \frac{\partial g(x;\beta)}{\partial \beta} \right)^T
M^{-1}
\mathbbm{E} \left( \varphi \varphi^T \right) 
(M^T)^{-1}
\left( \frac{\partial g(x;\beta)}{\partial \beta} \right) \right),
\end{equation}
\end{theorem}

Condition 1 of Theorem \ref{thm:largesample} requires that the influence function not be too complex (i.e., it must lie in a Donsker class). This is therefore an implicit complexity restriction on the nuisance functions $(\pi,\mu_0,\mu_1)$ as well. Donsker classes are described in detail in \citet{vandervaart} and \citet{kennedy2016}, but include for example smooth functions with bounded partial derivatives. However, the Donsker part of Condition 1 can be avoided entirely  with sample splitting, i.e., estimating the nuisance functions on one part of the sample and evaluating the influence function on a separate independent part. This was recently termed cross-fitting by \citet{chernozhukov}. We use this sample-splitting procedure in our simulations and application. Condition 2 is a weak differentiability condition that essentially just requires a smooth enough projection model $g(x;\beta)$.

Theorem \ref{thm:largesample} thus tells us that the rate of convergence of our proposed estimator $\hat\beta$ only has a weak second-order dependence on the errors of the nuisance estimators. In particular, this means $\hat\beta$ will be $\sqrt{n}$-consistent and asymptotically normal even if the nuisance estimators only converge at $n^{-1/4}$ rates. These rates can be attained in nonparametric models under sufficient smoothness, sparsity, or structural (e.g., generalized additive model) conditions.

 Let $\sigma^2(x)$ be shorthand for the asymptotic variance of $g(x, \hat\beta)$ given in Theorem \ref{thm:largesample}, so that under the given conditions we have
\begin{equation}
\label{eq:sandwichvariance1}
\sqrt{n} \left( g(x; \hat\beta) - g(x; \beta) \right) \indist N ( 0, \sigma^2(x)).
\end{equation}
Then a simple Wald-style 95\% confidence interval for $g(x;\hat\beta)$ is given by
\begin{equation}
\label{eq:confidenceintervals}
\left[ g(x;\beta) - 1.96 \frac{\hat\sigma(x)}{\sqrt{n}} , g(x;\beta) + 1.96 \frac{\hat\sigma(x)}{\sqrt{n}}  \right],
\end{equation}
where $\hat\sigma^2(x)$ is any consistent estimator of the asymptotic variance, e.g.,
$$ \hat\sigma^2(x) = \left( \frac{\partial g(x;\hat\beta)}{\partial \beta} \right)^T
\hat{M}^{-1}
\mathbbm{P}_n \left( \hat\varphi \hat\varphi^T \right) 
(\hat{M}^T)^{-1}
\left( \frac{\partial g(x; \hat\beta)}{\partial \beta} \right) $$
where $\hat{M} = \mathbbm{P}_n \Big(\frac{\partial \hat\varphi}{\partial \beta}\Big)$ is an estimate of the derivative matrix.

Importantly, if we had not used an influence function-based estimator for $\beta$, the corresponding rates of convergence would not in general involve products of the nuisance errors, so one would expect slower-than-$\sqrt{n}$ rates for estimating $\beta$ and no available confidence intervals.

\subsection{Model selection}

The conclusions of this estimation method will depend on the choice of parametric model $g(x;\beta)$. The user can select a model $g(x;\beta)$ by drawing on substance matter knowledge, or if the user prefers to follow a data-driven approach, he or she could use model selection by using cross-validation. In this section we describe the approach first proposed by \cite{laandudoit2003} and later by \cite{kennedylorchsmall2017} for model selection in the presence of nuisance parameters. 

Using cross-validation in this setting is slightly different than usual because the risk parameter depends on nuisance functions. Following a process similar to the one in  \cite{kennedylorchsmall2017}, we treat the risk as a parameter and derive its efficient influence function.

To rank a set of candidate estimators $\{\hat{g}_k: k \in \mathcal{K}\}$, we define the pseudo-risk for a candidate model $g_k(x;\beta)$ as
\begin{equation}
\label{eq:pseudo-risk}
    R(\hat{g}_k(x;\beta)) = \int_\mathcal{X} w(x) \{ \hat{g}_k(x;\beta)^2 - 2 \gamma(x) \hat{g}_k(x;\beta)\} dP(x),
\end{equation}
which is a shifted version of the mean squared error $R^*(\hat{g}_k(x;\beta)) = \int_\mathcal{X} w(x) \{ \gamma (x) - \hat{g}_k(x;\beta)\}^2 dP(x)$. The pseudo-risk will yield the same rankings of models as the true risk, since the pseudo-risk is a shifted version of the true risk, where the shift does not depend on the estimator $\hat{g}_k(x;\beta)$. The efficient influence function for the risk $R({g})$ of a given fixed candidate $g(x;\beta)$ is
\begin{align} \label{eq:riskif}
\varphi_R (X) = & w(X) \bigg\{ 2g(X;\beta) \bigg( \frac{\mu_0(X)}{\mu_1(X)} \frac{A(Y-\mu_1(X))}{\pi(X)} - \frac{1}{\mu_1(X)} \frac{1-A}{1-\pi(X)} (Y-\mu_0(X)) \bigg) \nonumber \\
& + g(X;\beta)^2 - 2 \gamma(X) g(X;\beta) \bigg\} - R(\hat{g}_k(X;\beta)).
\end{align}
\normalsize

A proof of this follows similar logic as the proof for Theorem 1. One can implement an influence function-based estimator of the pseudo-risk using \eqref{eq:riskif} and the same approach as in \ref{sec:ifs}. The resulting pseudo-risk estimates (over any given class of models) can then be used to select the best model, functional form, and covariates. One could also consider different risk measures, other than mean squared error, but this is beyond the scope of this article.

\subsection{Example: Projecting onto a logistic regression}

When using the proposed estimator for a dataset, suppose we select $g(x;\beta)=$logit($p$). The reader might find it puzzling that we now discuss PC as a logistic model, but recall that it is only the model onto which we are projecting the true function. The consequence of selecting a logit is that we obtain estimates of the outcome, which correspond to PC, and estimated coefficients, just like one does when performing logistic regression, which correspond to the relationship between the covariates and PC.

Specifically, how should these results be interpreted? A logistic regression of $Y$ on $X_1, \dots, X_k$ estimates parameter values for $\beta_0, \dots, \beta_k$ via maximum likelihood estimation of the following equation:
\begin{equation}
\label{eq:parammodel}
\text{logit}(p) = \log \left(\frac{p}{1-p}\right) = \beta_0 + \beta_1 X_1 + \dots + \beta_k X_k.
\end{equation}
In our example, $p$ is the probability of causation.

For a binary covariate $X_1$, the coefficient $\beta_1$ is the log-odds ratio between the group where $X_1=0$ and the group where $X_1=1$. As is usual, to translate to odds one can just exponentiate the log-odds. In our case, the odds refers to the ``odds of causation.'' The odds of causation are
\begin{equation}\label{eq:oddsofcausation}
\text{Odds of causation} = \frac{P( Y^0=0 \mid Y=1, A=1, X=x )}{P( Y^0=1 \mid Y=1, A=1, X=x )}.
\end{equation}

If the odds of causation are equal to three, for example, then we can say that it is three times more likely that the outcome $Y$ was caused by exposure $A$ than not. 
The odds ratio corresponding to covariate $X_1$, for example, is
\begin{equation}
\text{Odds ratio} = \frac{\text{Odds}( Y^0=0 \mid Y=1, A=1, X_1=1, X_2, \dots, X_n )}{\text{Odds}( Y^0=0 \mid Y=1, A=1, X_1=0, X_2, \dots, X_n )}.
\end{equation}
If the odds ratio is equal to 0.17, then, for a female $X_1=1$ the odds that the outcome $Y$ was caused by the exposure $A$ are 0.17 times greater than for a male, $X_1=0$.


\section{Simulation: IFB-estimator vs. plugin estimator} \label{sec:simulation}

In this section, we compare the finite sample performance of our nonparametric influence function approach to that of the currently used plugin estimators, both the parametric and nonparametric ones, by simulation\footnote{\textbf{Note}: The code for the simulation can be found online at \url{https://github.com/mariacuellar/probabilityofcausation/blob/master/Simulation.R}.}. We generate our covariates by following the example from \cite{kangschafer}. We define a vector of four covariates $X = (X_1, X_2, X_3, X_4)$ as normal random variables with mean 0 and variance 1. To see how the estimators performed in a more complicated setting that did not follow simple normal distributions, we used the transformations in \cite{kangschafer} to define a new vector ${X}^*$ of transformed covariates,
\begin{align}
X_1^* &= \text{exp}(X_1/2) \nonumber \\
X_2^* &= X_2/(1 + (X_1)) + 10 \nonumber \\
X_3^* &= (X_1*X_3/25 + 0.6)^3 \nonumber \\
X_4^* &= (X_2 + X_4 + 20)^2.
\end{align}
\cite{kangschafer} originally defined these transformations for demonstrating that parametric assumptions led to biased results, despite having diagnostics plots that seemed to point to no assumption violations. The $X^*$ are simply a useful set of transformations of normal $X$ for our simulation. 

We let $ \gamma(x) =P(Y^0 = 0 \mid Y^1 = 1, A=1, X=x) = \frac{\mu_1(x) - \mu_0(x)}{\mu_1(x)}, $
where $\mu_a(x) = P(Y=1 \mid X=x, A=a)$. Consider the  factorization
\begin{equation}
\begin{gathered}
P(Y^1 = 1) = \beta \\
P(X=x \mid Y^1=y) = f(x \mid y) \\
P(A=1 \mid X=x, Y^1=y) = \pi(x) \\
 P(Y^0 = 1 \mid Y^1 = 1, A=a, X=x) = 1 - \gamma(x) \\
 P(Y^0 = 1 \mid Y^1 = 0, A=a, X=x) = 0,
\end{gathered}
\end{equation}
where we have
\begin{equation}
\begin{gathered}
P(A=1 \mid X=x, Y^1=y) = \pi(x) \\
P(Y^0 = 1 \mid Y^1 = 1, A=0, X=x) = P(Y^0 = 1 \mid Y^1 = 1, A=1, X=x) \\
P(Y^0 = 1 \mid Y^1 = 0, A=0, X=x) = P(Y^0 = 1 \mid Y^1 = 0, A=1, X=x),
\end{gathered}
\end{equation}
by no unobserved confounders, i.e., exchangeability, $(Y^0,Y^1) \indep A \mid X$, and we have $ P(Y^0 = 1 \mid Y^1 = 0, A=a, X=x) = 0 $ by monotonicity $Y^0 \leq Y^1$. Then,
\begin{align}
\mu_1(x) &= P(Y=1 \mid X=x, A=1) = P(Y^1 \mid X=x) \\
&= \frac{ P(X=x \mid Y^1 = 1) P(Y^1=1) }{ \sum_y P(X=x \mid Y^1 = y) P(Y^1=y) } \\
&= \frac{ \beta f(x \mid 1) }{ \beta f(x \mid 1) + (1-\beta) f(x \mid 0) }
\end{align}
where the first equality follows by exchangeability and consistency, the second by Bayes rule, and the third by definition according to our factorization. Similarly using the fact that $\gamma = 1 - (\mu_0/\mu_1)$, 
\begin{align}
\mu_0(x) &= P(Y=1 \mid X=x, A=0) = (1-\gamma(x)) \mu_1(x) = \frac{ (1-\gamma(x)) \beta f(x \mid 1) }{ \beta f(x \mid 1) + (1-\beta) f(x \mid 0) } .
\end{align}
Now, for example, we take $\gamma(x) = \text{expit}(\psi^T X)$, so $\gamma$ is a logistic regression, and $f(x \mid y) = f(x)$ (so $X \indep Y^1$). Then, the functional forms of the outcome regressions are
\begin{equation}
\mu_1(x) = \beta \ \text{ and } \ \mu_0(x) =\frac{ \beta}{ 1 + \exp(\psi^T X)}.
\end{equation}
Note that $\mu_a(x)$ does not follow a logistic model. Following \cite{kangschafer}, let $\psi = (-1, 0.5, -0.25, -0.1)$. To estimate the variance of the influence-function-based estimator, we use the variance from \eqref{eq:sandwichvariance}. As an alternative to the sandwich variance, we could use the bootstrap, since our estimator is asymptotically normal, although the bootstrap is computationally much slower.

\begin{figure}[ht]
  \centering
  \includegraphics[width=\textwidth]{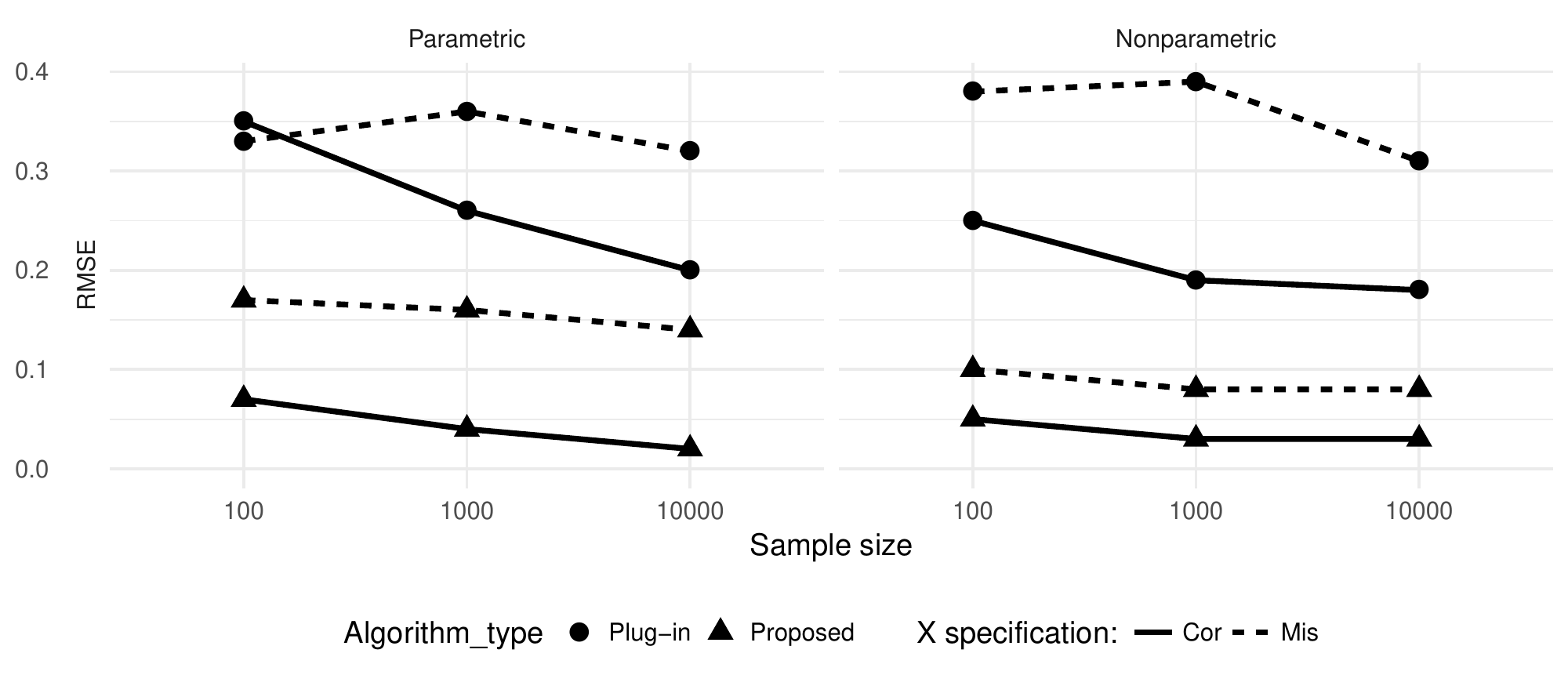} 
  \caption{Root-mean-squared-error (RMSE) across 100 simulations. The RMSE is higher for the plugin than for the proposed IFB estimator, and the misspecified model (or $X^*$ in the nonparametric case) has a slower decrease of error over sample size. See Table 1 for full results.} 
  \label{fig:RMSEsimulation}
\end{figure}

The purpose of the simulation was to test the convergence rates of the proposed estimator versus those of the plugin. Thus, we have set $\gamma = g$, whereas in reality the IFB method estimates a projection of the true function onto a parametric model (they are not equal, but instead $\gamma$ is projected onto $g$). Note that our simulation assumes $g=\gamma$, and thus we are not testing how projecting onto different forms of $g$ affects the estimates. Instead, we chose to focus on demonstrating the superiority of the IFB estimator over the plugin. 

With $J=100$ simulations, we assessed the performance of the estimators by using the root-mean-squared error (RMSE), absolute bias, and coverage:
\begin{align}
& \text{RMSE of }g(x;\hat\beta): \sqrt{  \frac{1}{J} \sum_{j=1}^J \left( g_j (x;\hat\beta) - g_j (x;\beta) \right)^2 }\\
& \text{Mean absolute bias of }g(x, \hat\beta): \frac{1}{J} \sum_{j=1}^J \left| ( g_j (x;\hat\beta) - g_j (x;\beta) ) \right| \\
& \text{Coverage: Proportion of simulations in which the confidence interval covered the true value.}
\end{align} 
Since our estimates for $g$ depend on $x$, we chose to evaluate the bias an RMSE for each of 2,000 simulated $X$ covariates, and then take the average of those to obtain one bias and one RMSE per simulation. To compare the performance of the estimators across sample sizes, we tested our method at three sample sizes: 100, 1,000, and 10,000. We used random forests to fit the nuisance parameters, but this could also be done with other nonparametric estimators.

\begin{table}[t]
\centering
\caption{Bias (RMSE) across 100 simulations. GLM and random forests were used for parametric and nonparametric estimation of the nuisance parameters.}
\begin{tabular}{cccccc}
\hline
Sample size & Method & Parametric & Parametric & Nonparametric & Nonparametric  \\
 &  & correctly specified ($X$) & misspecified ($X^*$) & with $X$ & with $X^*$ \\
 \hline
100 & Plugin & 0.36 (0.35) & 0.25 (0.33) & 0.09 (0.25) & 0.37 (0.38) \\ 
  & Proposed & 0.01 (0.07) & 0.04 (0.17) & 0.01 (0.05) & 0.02 (0.1) \\ 
 1000 & Plugin & 0.19 (0.26) & 0.22 (0.36) & 0.05 (0.19) & 0.34 (0.39) \\ 
  & Proposed & 0.00 (0.04) & 0.04 (0.16) & 0.00 (0.03) & 0.01 (0.08) \\ 
 10000 & Plugin & 0.07 (0.20) & 0.16 (0.32) & 0.06 (0.18) & 0.14 (0.31) \\ 
  & Proposed & 0.00 (0.02) & 0.03 (0.14) & 0.00 (0.03) & 0.01 (0.08) \\ 
\hline
\end{tabular}
\label{tab:simresults}
\end{table}

Figure \ref{fig:RMSEsimulation} (with more complete values given in Table \ref{tab:simresults}) shows the motivation for using our nonparametric IFB approach. If the model is correctly specified, the IFB estimates have a greater variance in terms of a difference in a constant than the plugin, but they still have root-n rate and inference under weak conditions. But misspecified parametric models for the nuisance functions are more common. In this case, the error of the plugin estimator does not decrease with sample size. Instead, there is an irreducible amount of bias the prevents the estimates from reaching the true value. With correctly specified parametric models for the nuisance functions, we expect our proposed IFB estimator to perform about as well as the correctly specified parametric plugin model, although in this case it performs better. With misspecified models for the nuisance functions, the IFB estimator performs better than the parametric plugin under misspecified models. Furthermore, if the model is misspecified, the IFB estimator estimates a well-defined quantity.

With nonparametric nuisance functions, the error of the plugin estimates decreases with increasing sample size, even when using the transformed covariates $X^*$. Since one can never guarantee one knows the true parametric model, it makes sense from this argument that one should want to use nonparametric models for the nuisance functions. However, nonparametric methods do not usually allow for the derivation of valid confidence intervals. This is why our proposed estimator is useful. The proposed IFB estimator has faster convergence rates than the nonparametric plugin, \emph{and} it allows for the derivation of confidence intervals under weak conditions.

Finally, the coverage of the confidence intervals for the IFB estimates are shown in Table \ref{tab:simresultscoverage}. The IFB estimator has close to 95\% coverage in the correctly specified parametric setting (which we would not expect to have in most cases) and in the nonparametric cases, although it has higher coverage with the simpler covariates. In the misspecified parametric setting it does not have high coverage, but this is expected because it is precisely the result of misspecification. We suggest using the nonparametric approach instead.

\begin{table}[t]
\centering
\caption{Coverage across 100 simulations using nonparametric estimation (random forests) for nuisance parameters. Full definition of terms given in Table 2.} 
\begin{tabular}{cccccc}
\hline
Sample size & Parametric & Parametric & Nonparametric & Nonparametric  \\
& correctly specified ($X$) & misspecified ($X^*$) & with $X$ & with $X^*$ \\
\hline
 100 & 92.10 & 80.82 & 94.22 & 90.01 \\ 
 1,000 & 94.42 & 75.23 & 94.50 & 92.21 \\ 
 10,000 & 95.19 & 69.35 & 95.23 & 93.08 \\ 
 \hline
\end{tabular}
\label{tab:simresultscoverage}
\end{table}


\section{Application}\label{sec:application}

\subsection{Data}

We used the dataset that was carefully gathered by \cite{jpal}, as part of a project from the Abdul Latif Jameel Poverty Action Lab (JPAL). \footnote{The dataset, documentation, and article about this study are available online at \url{https://www.povertyactionlab.org/evaluation/cleaning-springs-kenya}. The materials are kept in the Harvard Dataverse and are currently maintained by Professor Edward Miguel. Code for the application can be found online at \url{https://github.com/mariacuellar/probabilityofcausation/blob/master/Application}.} We applied our method to estimate the already mentioned probability of causation, 
\begin{equation}
PC(x) = P(Y^0=0 | Y=1, A=1, X=x),
\end{equation}
where $A=1$ is the exposure to high bacterial concentration (unprotected) springs, $A=0$ is the exposure to low bacterial concentration (protected) springs, $Y=1$ is the outcome that the household had a child with diarrhea within the past week, $Y=0$ is the outcome that no child had diarrhea within the past week, and $X$ is the vector of covariates for the households. These include: the child's gender (0=male, 1=female), the child's age at baseline, the mother's years of education at baseline, the water quality at baseline in terms of E.coli MPN (most probable number), whether the home as an iron roof (which could prevent some water cross-contamination at the home), the mother's hygiene knowledge at baseline (determined from a survey about prevention of diarrheal diseases), the latrine density near the house at baseline (how many latrines are near the house per unit area, which could affect the fecal matter concentration near springs), and the number of children in the household at baseline.

For our analysis, we restricted the sample in the same way as the \cite{jpal} researchers. We removed all individuals older than age three, all the children of reported users of multiple springs, and all the children flagged as having a problem weight, problem BMI or being a severe height outlier (as defined by the researchers). This reduced the sample size from 45,565 to 22,620. Furthermore, we removed all observations containing any missing values, like \cite{jpal} do, which reduced the sample size to 2,933. The drop in sample size from 45,565 to 2,933 is concerning, and a more careful view of the data would have to take the missing observations into account. For now, we assume that the values were missing at random, and understanding the consequences this assumption further remains as future work.

\begin{figure}[t]
  \centering
  \includegraphics[width=\textwidth]{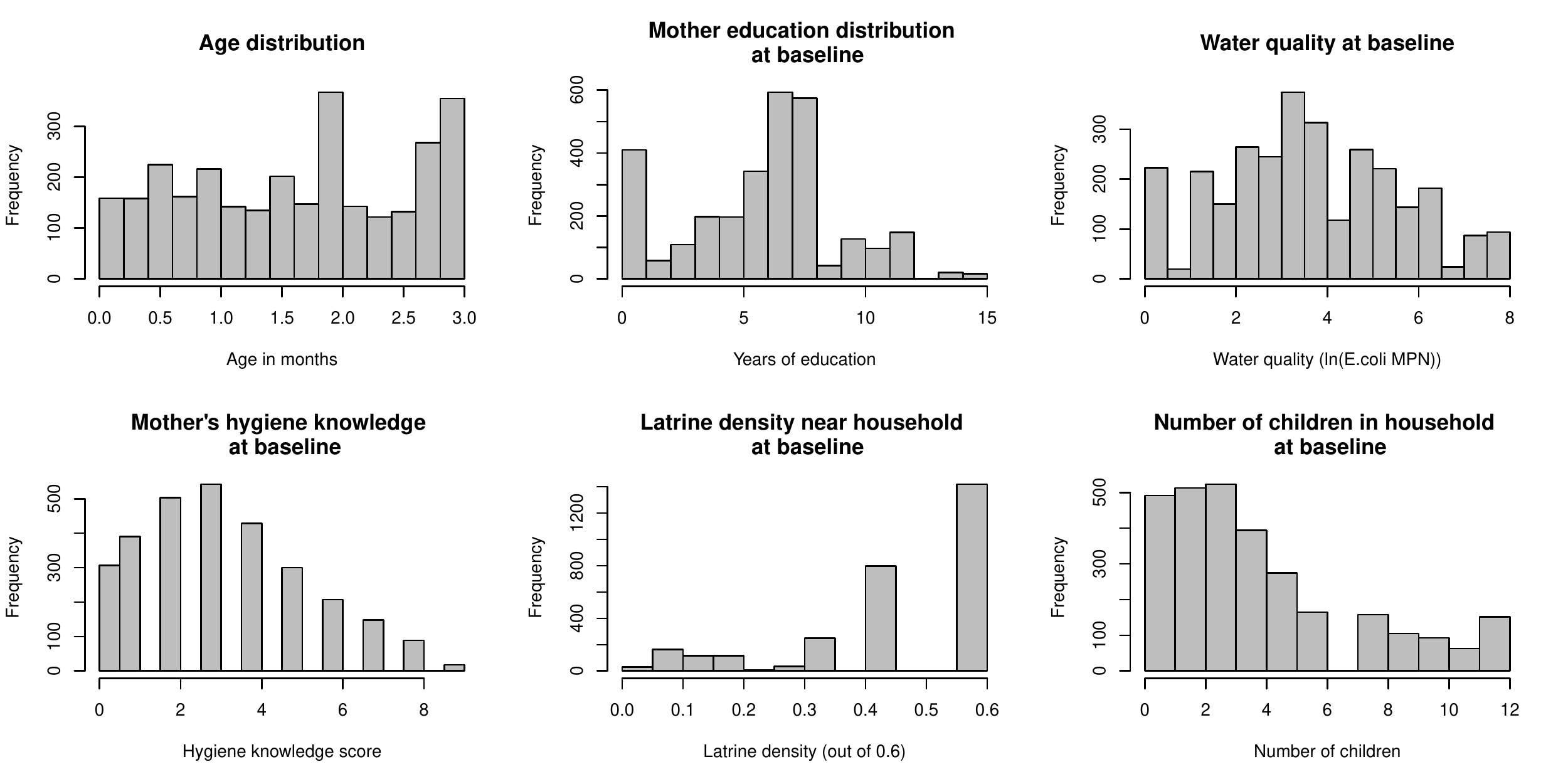} 
  \caption{Distributions of variables for the households in the selected sample for testing the IFB method.}
  \label{fig:EDA}
\end{figure}

After the data reduction, the sample had 1,446 children who drank from unprotected springs ($A=1$) and 1,487 children who drank from protected springs ($A=0$). There were 688 children with diarrheal diseases in the past week ($Y=1$), and 2,245 ($Y=0$) children without the diseases in the past week. There were 1,457 boys and 1,476 girls. 849 of the households did not have an iron roof, and 2,084 did. Figure \ref{fig:EDA} shows the histograms of the remaining six variables.

The researchers designed the experiment such that the following assumptions were likely to be satisfied:
\begin{enumerate}
\item \textbf{Positivity}: The researchers selected 400 water springs that were all suitable for treatment.
\item \textbf{Consistency}: The researchers tested for effects of interference by measuring the bacterial levels in treated and untreated springs that were near each other (within 3 km), and found ``little evidence of externalities in water quality.'' (\cite{jpal}) They concluded that it was unlikely that treating one spring had an effect on the bacterial concentration of an untreated spring. No interference, along with a well-defined binary treatment, means consistency is also satisfied.
\item \textbf{Monotonicity}: The researchers did not expect that protecting a spring would increase the bacterial concentration and thus increase the incidence of diarrhea in children.
\item \textbf{No unobserved confounders}: The treatment was given as a random assignment over the sites. The researchers checked that households were balanced in terms of demographic characteristics (such as income, education and health), and that the individuals did not use additional forms of bacterial reduction (such as home water chlorination, boiling or hygiene practices).
\end{enumerate}

\subsection{Results}

\subsubsection{Result 1: Estimated PC and confidence intervals for a single individual}\label{sec:estpcindiv}

\begin{table}[t]
\centering
\caption{Estimates of the probability of causation and the corresponding 95\% confidence intervals, for a ``median'' child with covariates given in 5.2.1.} 
\begin{tabular}{ lcc } 
\hline
 & \multicolumn{2}{c}{Probability of causation} \\ \cline{2-3} 
& Estimate & Confidence interval \\ 
\hline 
 Parametric plugin (logistic regressions) & 0.69 & (0.65, 0.73) \\
 Nonparametric plugin (random forests) & 0.24 & (undefined) \\ 
 Nonparametric proposed (random forests) & 0.12 & (0.11, 0.13) \\ 
 \hline
\end{tabular}
\label{tab:estimatedPCx}
\end{table}

Table \ref{tab:estimatedPCx} shows our results from estimated $PC(x)$ and its confidence intervals for an individual with specific covariates as an example, by using the currently-used methods and our proposed method. Recall that we are interested in the probability of causation: the probability that it was the exposure to high-bacterial springs, and not something else (such as transfer at school or a stomach virus), that caused the child diarrhea in the community of interest.

The covariates selected for Table \ref{tab:estimatedPCx} are the mode (for binary variables) and median (for continuous) variables of the data---we call this child the ``median'' child: A girl, aged six, whose mother has six years of education, with poor quality water at baseline (level 4), living in a house with an iron roof, baseline hygiene knowledge of 3, latrine density around nearest spring of 0.4 out of 0.6, with five children in the household. The covariates used here were selected for being the most common values for each variable as an example, but PC could be estimated for other covariates as well. We indeed calculated PC for a variety of covariates. Some of the variability of PC can be see in Result 3.

We used random forests to estimate the nuisance functions $\hat\eta=(\hat\pi, \hat\mu_0, \hat\mu_1)$. According to our estimator, the probability that the child diarrheal disease was caused by the exposure to dirty unprotected springs is 0.12 with a 95\% confidence interval of (0.11, 0.13). The parametric plugin method, which is the most commonly used, yielded a dramatically higher estimate of 0.69 with a confidence interval (0.65, 0.73), and the nonparametric plugin method yielded an estimate of 0.24.

The plugin estimates were dramatically different from the IFB estimates, indicating that the parametric plugin model might be highly misspecified. In other words, although the currently used method suggests that it is likely that the diarrheal disease was caused by the exposure to bacteria, once we remove the parametric assumptions and use our method, we find the opposite: that it is not likely that the diarrheal disease was caused by the exposure to bacteria.

\subsubsection{Result 2: The odds of causation} 

\begin{table}[t]
\centering
\caption{Regression results. Estimates are estimated coefficients of model shown in (31), which are the log-odds of causation shown in (32).} 
\begin{tabular}{lccccc}
\hline
 & Estimate & Robust std. error & z value & p value & \\ 
 \hline
Gender (male:0, female:1) & 0.17 & 0.09 & 1.99 & 0.05 & \\ 
 Age (in months) & -0.21 & 0.04 & -4.93 & 0.00 & *** \\  
 Mother's years of educ. & -0.03 & 0.01 & -2.07 & 0.04 & * \\ 
 Water quality at spring & -0.06 & 0.02 & -2.81 & 0.00 & *** \\ 
 Iron roof indicator & -0.05 & 0.09 & -0.57 & 0.57 & \\ 
 Mother's hygiene knowledge & -0.08 & 0.02 & -4.01 & 0.00 & *** \\ 
 Latrine density near household & -0.28 & 0.24 & -1.15 & 0.25 & * \\ 
 Diarrhea prevention score & -0.02 & 0.02 & -1.15 & 0.25 & * \\ 
 \hline
\end{tabular}
\\
Sample size: 2,933 observations. Significance codes: 0 `***' 0.001 `**' 0.01 `*' 0.05 `.' 0.1 ` ' 1. Variables are given at experiment baseline.
\label{tab:coefficientsinterpret}
\end{table}

Table \ref{tab:coefficientsinterpret} shows the estimated odds ratios under the column header ``Estimate''. For example, the coefficient on latrine density near household is -0.28 (with significance at the 0.01 level). This means that for a one-unit increase in the density, the expected change in log-odds of causation is 0.28 and the change in odds is $e^{(0.28)}=1.32$. So we can say for a one-unit increase in the density, we expect to see about 32\% increase in the odds of causation, i.e., the probability that the child's diarrhea was caused by exposure to the bacteria increases as latrine density increases. 

The factors that most affect whether the child's diarrheal disease was caused by his or her exposure to bacteria in the water springs are: child's age and latrine density near household. The mother's years of education, water quality at spring, mother's hygiene knowledge, and diarrhea prevention score all had a smaller effect, but they were significant at least at the 0.05 level.

This information can be used to learn which groups of people were actually harmed from being exposed to the bacteria in the water springs. The individuals with high odds of causation are the ones who are most harmed by the exposure. Policy makers might want to target these specific subgroups specifically, especially given a limited budget.

\subsubsection{Result 3: Variation in PC for a given continuous covariate}

\begin{figure}[t]
\centering
\begin{minipage}{.45\textwidth}
 \centering
 \includegraphics[width=\linewidth]{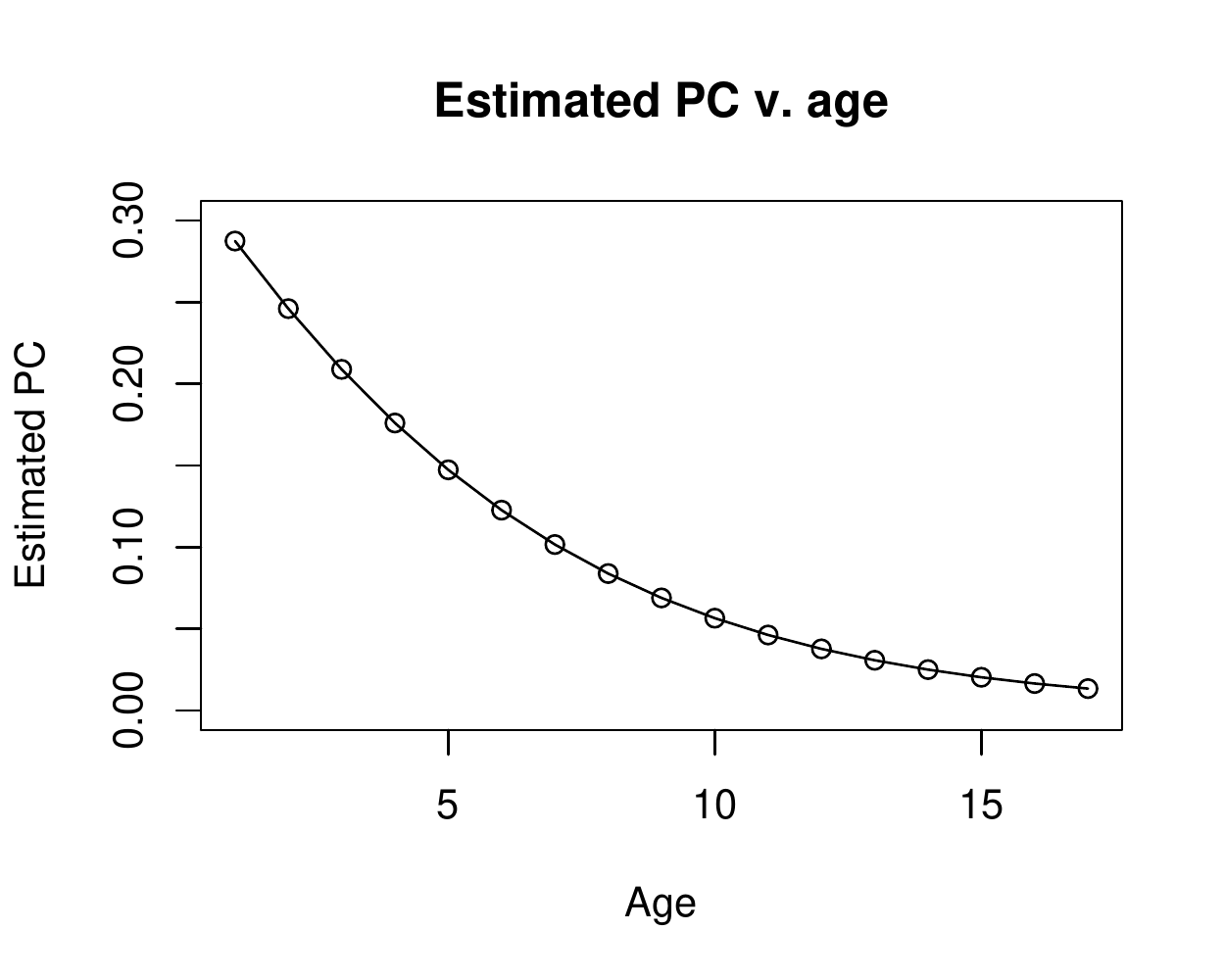}
 \captionof{figure}{Variation of the estimated probability of causation with changes in the child's age, ceteris paribus.}
 \label{fig:pcvsage}
\end{minipage} \quad \quad
\begin{minipage}{.45\textwidth}
 \centering
 \includegraphics[width=\linewidth]{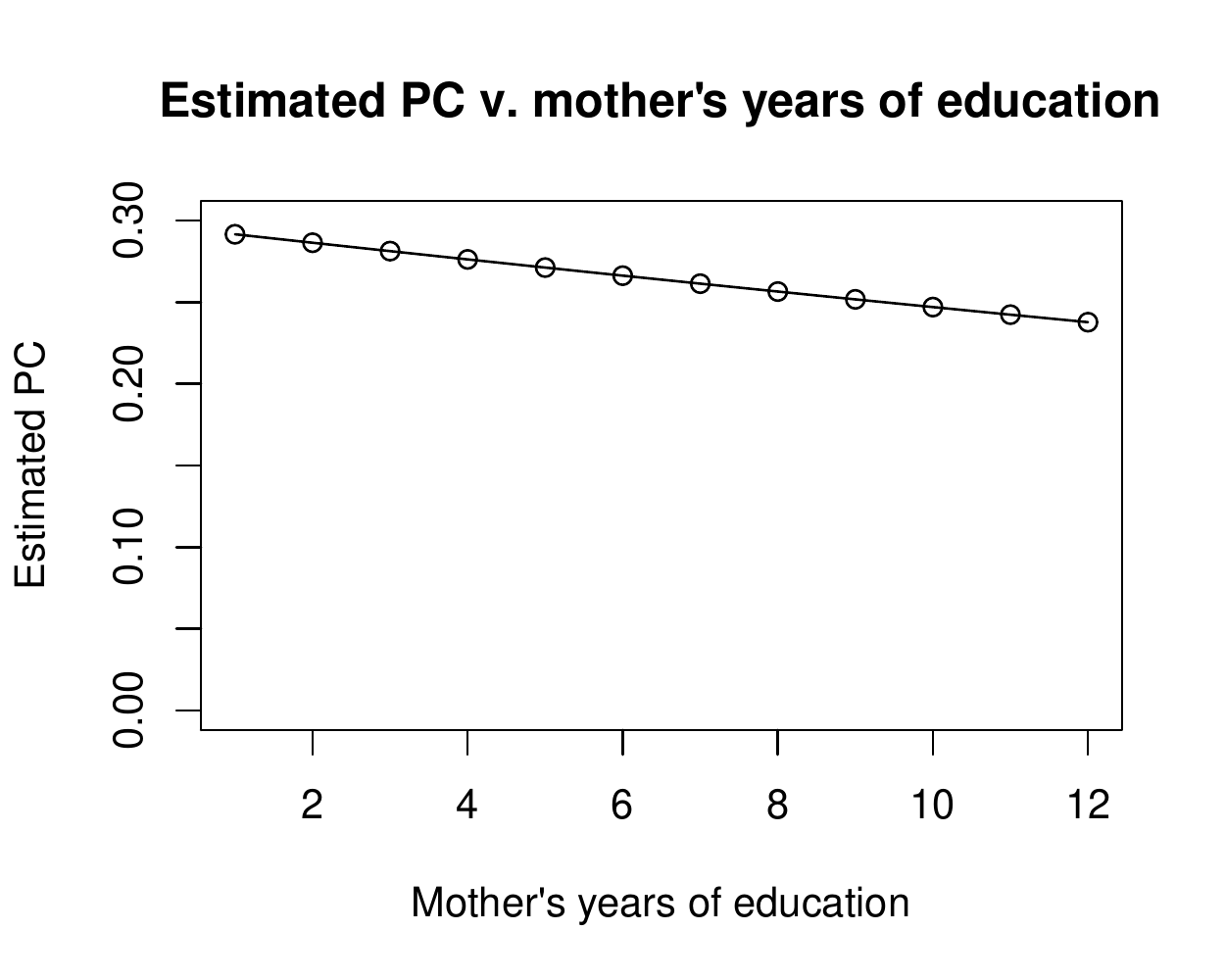}
 \captionof{figure}{Variation of the estimated probability of causation with changes in the mother's years of education, ceteris paribus.}
 \label{fig:pcvsmothereduc}
\end{minipage}
\end{figure}

We now estimate PC for an individual with the same covariates as in Result 1, except one covariate is varied to see how the estimated parameter changes. Figure \ref{fig:pcvsage} shows the variation of PC with age. As the child's age increases, the probability that the diarrheal disease was caused by a bacterial exposure decreases non-linearly. Figure \ref{fig:pcvsmothereduc} shows the variation of PC with the mother's years of education. We see that the more educated the mother, the less likely it is that the child's diarrheal disease was caused by the bacterial exposure. 

This exercise could be repeated for whatever variables might be of interest. Thus, it is an exercise in exploring the heterogeneity in estimates of PC across the space of covariates. Finding out which groups of individuals have higher estimates of PC could be very useful for researchers when deciding on which sub-populations to implement costly interventions.

Note the difference between studying heterogeneity in the average treatment effect (ATE) over different subgroups of the population, and studying the probability of causation (PC). The ATE shows for which groups the exposure causes an outcome, on average. The PC shows for which groups it was the exposure, and not something else, that actually caused an outcome. PC is more useful in reviewing what went wrong in specific cases (as one would do in a class action lawsuit, for instance), and it could help understand when there might be other causes of the outcome of interest.


\section{Discussion}\label{sec:discussion}

We derived a novel estimator for the probability of causation. This influence-function-based estimator is for a projection of the identified probability of causation, which does not require making parametric assumptions. This approach has applications in diverse settings, including epidemiology and the law.

The main contribution of our work was providing an estimator for the probability of causation without requiring strong parametric modeling assumptions by using influence functions. We provided an expression for the efficient influence function that involves only derivatives based on the known (and generally simpler) model and weight functions, rather than unknown complex and high-dimensional nuisance functions. Our formulation yields efficient estimators that are easier to construct in practice. We also described the asymptotic properties of our method under weak empirical process conditions and provided simulations to compare it to parametric and nonparametric plugin methods.

Regarding the projection, what is the difference between assuming the data-generating process actually follows an incorrect parametric model ($\gamma(x)=g(x;\beta)$) and using a projection approach? If the model $g(x;\beta)$ happens to be correct, then both the projection approach and a model-based approach will in general yield $\sqrt{n}$-consistent estimators with valid confidence intervals. This is true under roughly the same assumptions, though the projection approach may have a larger variance, but only by constant factors. On the other hand, if the model $g(x;\beta)$ is not correct, then the projection approach is honest about this possibility and still validly defined as a best-fitting approximation, while the parametric model-based approach is no longer formally applicable. In other words, regardless of whether a modeling assumption is correct, the projection will yield correct estimators, whereas assuming the data follow a model could yield biased estimates.

Regarding the choice of parametric model onto which we project the true data, what is the proper way to choose $g(x;\beta)$? After all, selecting different models could yield different estimates. How sensitive the results are to the choice of model depends on how complex the true function is. \cite{kennedylorchsmall2017}, section 3.4, wrote about how to do model selection using cross-validation in a projection setting for a continuous instrumental variable. They noted that in standard cross-validation it is possible to estimate the risk without bias since the risk does not require the estimation of nuisance parameters. However, in their setting, as in ours, the risk parameter depends on complex nuisance functions via the true parameter of interest and the parametric model being used. They derive an efficient influence function for the risk and then use this as a doubly robust loss function for cross-validation-based model selection. The results in this paper pave the way for a similar data-driven model selection criterion, but we leave the cross-validation type approach for selecting $g(x;\beta)$ as future work and simply use a logit model as our projection for illustration.

In our application with real-world data, we used the proposed method to study whether exposure to high-bacteria water is what caused diarrheal disease in children from a population in Kenya, or if the cause was something else. We estimated the probability of causation by using a dataset from an experiment performed by JPAL. We found that the plugin estimates were dramatically different from the IFB estimates, indicating that the parametric plugin model might be highly misspecified. This difference is likely a product of the misspecified plugin model and not our choice of projection model. Our method suggests the opposite: Because PC is low, it is unlikely that the diarrheal disease was caused by the exposure to the high concentration of bacteria in the drinking water. We also found that different populations had different values of PC, that is, there was heterogeneity between groups of individuals. This is an interesting result insofar as it suggests that although some individuals developed the outcome, it is not necessarily the case that it was caused by the exposure, and thus that there might be other causes of this childhood disease. It would be valuable to perform a further medical evaluation to evaluate which groups had higher PC, since future interventions might be more effective for those groups.

The factors that most affect whether the child's diarrheal disease was caused by his or her exposure to bacteria in the water springs are: age, the mother's years of education, water quality at spring, mother's hygiene knowledge, latrine density near household, and diarrhea prevention score. This information could be used to learn for which groups of people the risk of harm from being exposed to the bacteria in the water springs is high. The individuals with high odds of causation are the ones who are most harmed by the exposure. Policy makers might want to target these specific subgroups.

Apart from providing useful information about this specific population in Western Kenya, this application served as an example of how one might estimate and interpret the probability of causation, as well as the estimated coefficients of a logistic regression representing PC, the odds of causation. The method derived in this article provides a novel way of analyzing data from randomized trials (and other datasets following the assumptions we listed), and thus it could be used to find new insights about causation and public policy. Further work needs to be done to relax exchangeability, or no unobserved confounders. Though this assumption is satisfied by randomized trials, it is an unrealistic assumption for some applications. Nevertheless, estimating PC is appropriate whenever researchers are interested in finding out whether a specific intervention had the effect intended by the policymaker.



\vspace{-.1in}
\section*{Conflict of interest}
The authors have no conflict of interest to report.

\bibliographystyle{rss.bst}
\bibliography{pcausation.bib}


\section{Appendix}

\subsection{Proof of Theorem 1}

The first step is deriving the influence functions of the nuisance parameters $\eta(x) = ( \mu_0(x), \mu_1(x), \pi(x) )$ in the simple discrete case. For this task, we cite a useful result. 

\begin{lemma}
When $X$ is discrete, the efficient influence function (EIF) for the parameter $\mathbbm{E}(Y \mid X=x)$ is given by
\begin{equation}\label{eq:edsrule}
\varphi(Z) = \frac{\mathbbm{1}( X = x )}{P(X=x)} (Y - \mathbbm{E} (Y \mid X=x) ),
\end{equation}
where $Z=(X,Y)$.
\end{lemma} 

\begin{proof}
A simple proof of the above follows from the heuristic that influence functions are derivatives, and so usual derivative rules apply. In particular since $\mathbbm{E} (Y \mid X=x) = \mathbbm{E}(Y \mathbbm{1} (X=x) )/P(X=x),$ then by the quotient rule, the EIF of the ratio is
\begin{equation}
\frac{Y \mathbbm{1}(X=x) - \mathbbm{E} (Y \mathbbm{1}(X=x))}{P(X=x)} - \frac{\mathbbm{E}(Y \mathbbm{1}(X=x) ) (\mathbbm{1}(X=x)- P(X=x))}{P(X=x)^2} ,
\end{equation}
which simplifies to $\frac{\mathbbm{1}(X=x) }{P(X=x)} (Y - \mathbbm{E} (Y \mid X=x) = \varphi(Z)$, the desired result.
\end{proof}

Continuing our heuristic of working in the simple discrete case, we seek to derive the influence function $\varphi(Z)$ of the moment condition parameter at a given $\beta$
\begin{equation} \label{eq:momentcon}
  \Psi(\beta) = \sum_x h(t) \{ \gamma(t) - g(t;\beta) \} p(t),
\end{equation}
where $h(X) := \frac{\partial g(X;\beta)}{\partial \beta} w(X)$. 
When we derive the EIF in the discrete case we will be able to see the continuous generalization.

We use the notation $IF(\bullet)$ to denote the influence function of a quantity $\bullet$. The influence functions of $\mu_1(x)$, $\mu_0(x)$, and $P(x)$ are
\vspace{-.2in}
\begin{align}
IF(\mu_0(x)) & = IF ( \mathbbm{E}( Y \mid A=0, X=x ) ) = \frac{\mathbbm{1}(A=0)\mathbbm{1}(X=x)(Y-\mu_0(x))}{(1-\pi(x))p(x)} ,\\
IF(\mu_1(x)) & = IF( \mathbbm{E}( Y \mid A=1, X=x ) ) = \frac{\mathbbm{1}(A=1)\mathbbm{1}(X=x)(Y - \mu_1(x))}{\pi(x) p(x)} , \\
IF(p(x)) & = IF ( \mathbbm{E} ( \mathbbm{1}(X=x) )) = \mathbbm{1}(X=x) - p(x). 
\end{align}

Recall the heuristic that influence functions are derivatives. Thus, we employ the chain rule on \eqref{eq:momentcon} to get
\begin{align}
IF(\Psi(\beta)) := \varphi(Z) = & \sum_x h(t) \left[ \{IF (\gamma(t)) - g(t;\beta)\}p(t) + \{\gamma(t) - g(t;\beta)\} IF (p) \right] \\
 = &\sum_x h(t) \{IF (\gamma(t)) - \gamma(t) \} p(t) + \gamma(t) - g(t;\beta). \label{eq:momentconexpanded}
\end{align}

To find $IF (\gamma(t))$, we use the quotient rule,
\begin{equation}
IF(\gamma(x)) = IF \left( 1-\frac{\mu_0(x)}{\mu_1(x)} \right) = \frac{\mu_0(x) IF(\mu_1(x)) - IF(\mu_0(x))\mu_1(x)}{\mu_1(x)^2}.
\end{equation}
We plug $IF(\gamma(x))$ into \eqref{eq:momentconexpanded},
\begin{equation}
\varphi(Z) = \sum_x h(t) \left\{ \left( \frac{\mu_0(t) IF(\mu_1(t)) - IF(\mu_0(t))\mu_1(t)}{\mu_1(t)^2} - \gamma(t) \right) p(t) + \gamma(t) - g(t;\beta) \right\},
\end{equation}
we substitute in $IF(\mu_0(x)), IF(\mu_1(x)),$ and $IF(p(x))$, and account for the fact that the sum $\sum_x$ combined with the indicator $\mathbbm{1}(X=x)$ ``picks out'' only the cases with $X$. Finally, we get that the efficient influence function (EIF) of the moment condition at a fixed $\beta$ is \small
\begin{equation}
\varphi(Z) = h({X}) \left\{ \frac{\mu_0({X})}{\mu_1({X})^2} \frac{A}{\pi(X)} \bigg( Y-\mu_1({X}) \bigg)
- \frac{1}{\mu_1({X})} \left( \frac{1-A}{1-\pi({X})} \right) \bigg(Y-\mu_0({X})\bigg) + \bigg( \gamma(X) - g(X;\beta) \bigg) \right\} - \Psi(\beta), 
\end{equation} \normalsize
as desired.

The result can be proved more formally by checking that  $\varphi(z)=\varphi(z;P)$ is a pathwise derivative in the sense that
$$ \frac{\partial\Psi(P_\epsilon)}{\partial \epsilon} \Bigm|_{\epsilon=0} = \int \varphi(z; P) \left( \frac{\partial \log dP_\epsilon}{\partial \epsilon} \Bigm|_{\epsilon=0}\right) dP $$
where $P_\epsilon$ is any smooth one-dimensional parametric submodel with $P_0=P$ (\cite{bickel,van2003unified,vandervaart}), and $\Psi(P_\epsilon)$ is the moment condition evaluated at such a submodel. Then the result follows since under a nonparametric model there is only one influence function and it is thus necessarily the efficiency influence function.

\subsection{Proof of Theorem 2}

Note that $\hat\Psi (Z;\hat\beta, \eta)= 0 $ by the definition of our estimator and $\Psi (X;\beta, \eta)=0$ by the definition of our parameter. Also, note that $\hat\Psi (Z;\hat\beta, \eta) = \mathbbm{P}_n \hat\varphi (Z;\hat\beta, \eta) $ and $\Psi (X;\beta, \eta) = P \varphi(Z;\beta,\eta)$ by the definition of influence functions. 

Thus, we can do the usual von Mises expansion on 
$$
0 = \hat\Psi (Z;\hat\beta, \eta) - \Psi (X;\beta, \eta)= \mathbbm{P}_n \hat\varphi (Z;\hat\beta, \eta) - P \varphi(Z;\beta,\eta)
$$ 
into four terms (where we have omitted the arguments on the second line for brevity),
\small \begin{align} 
0 & = \mathbbm{P}_n \hat\varphi (Z;\hat\beta, \eta) - P \varphi(Z;\beta,\eta) \nonumber \\
& = (\mathbbm{P}_n-P) (\hat\varphi - \varphi) + (\mathbbm{P}_n-P) \varphi + P(\hat\varphi - \varphi) \\
& = \underbrace{(\mathbbm{P}_n-P) (\hat\varphi - \varphi)}_{\text{I}} 
+ \underbrace{(\mathbbm{P}_n-P) \varphi }_{\text{II}} 
+ \underbrace{P (\hat \varphi(\hat\beta) - \varphi(\hat\beta) }_{\text{III}}
+ \underbrace{P (\varphi(\hat\beta) - \varphi(\beta)) }_{\text{IV}} . 
\label{eq:fourterms}
\end{align} \normalsize

\textbf{I)} Term I can converge to zero for two reasons. If we assume the true functions ${\pi}, {\mu}_0, {\mu}_1$ and their estimators $\hat{\pi}, \hat{\mu}_0, \hat{\mu}_1$ are in Donsker classes, and the functions in the class containing $\hat{\pi}$ are uniformly bounded away from zero and one, Term I is $o_\mathbbm{P}(n^{-1/2})$. Briefly, for the reader unfamiliar with empirical processes, classical empirical process theory deals with the empirical distribution function based on $n$ i.i.d. random variables. If $X_1, X_2, \dots, X_n$ are i.i.d. real-valued random variables with distribution function $P$. Donsker classes are important in empirical process theory. They are sets of functions with the property that empirical processes indexed by these classes converge weakly to a certain Gaussian process. This property allows us to say that Term I is $o_\mathbbm{P}(n^{-1/2})$. See \cite{kennedy2016}, Section 4.2, and \cite{vandervaart} Chapter 19, for a more in-depth explanation of the Donsker property and how it implies the convergence in distribution to zero of these terms. Alternatively, we can avoid Donsker classes altogether if we use data splitting, and $|| \eta(x) - \hat{\eta}(x) || = o_p(1)$ (\cite{kennedy2016}), which is a weak assumption. In this case, Term I will also be $o_\mathbbm{P}(n^{-1/2})$.

\textbf{II)} Term II converges to a normal distribution by the classical central limit theorem, as long as its mean and variance exist. This is because 
\begin{equation}
(\mathbbm{P}_n-P) \varphi (\beta, \eta) = \mathbbm{P}_n \varphi (\beta , \eta), 
\end{equation}
since $P\varphi (\beta, \eta)=0$ by definition of our parameter. $\mathbbm{P}_n \varphi (\beta, \eta)$ is a sample average.

\textbf{III)} Next, we want to show that Term III is $o_\mathbbm{P}(n^{-1/2})$. We will need iterated expectations, conditioning on $X$, conditioning on $A$, and using the fact that $\mathbbm{1}(A=1) \mu_A = A\mu_1.$ Our goal here is to show that $P (\hat \varphi(\hat\beta)$ is at most a second-order term:
\vspace{-.2in}
\begingroup
\allowdisplaybreaks
\begin{align*}
& P (\hat \varphi(\hat\beta) - \varphi(\hat\beta) )\nonumber \\
& =\frac{\hat{\mu}_0}{\hat{\mu}_1} \frac{\pi}{\hat{\pi}} \frac{(\mu_1 - \hat{\mu}_1)}{\hat{\mu}_1} - \frac{(1-\pi)}{1-\hat{\pi}} \frac{(\mu_0 - \hat{\mu}_0)}{\hat{\mu}_1} - ( \gamma - \hat{\gamma} ) \nonumber \\
& = (1-\hat{\gamma}) \frac{\pi}{\hat{\pi}} \frac{(\mu_1 - \hat{\mu}_1)}{\hat{\mu}_1} - \frac{(1-\pi)}{1-\hat{\pi}} \frac{(\mu_0 - \hat{\mu}_0)}{\hat{\mu}_1} - ( \gamma - \hat{\gamma} ) \nonumber \\
& = \underbrace{\left[ (1-\hat{\gamma})\frac{(\pi-\hat{\pi})}{\hat{\pi}} \frac{(\mu_1 - \hat{\mu}_1)}{\hat{\mu}_1} + \frac{(\pi-\hat{\pi})}{1-\hat{\pi}} \frac{(\mu_0 - \hat{\mu}_0)}{\hat{\mu}_1} \right]}_{\text{= Second-order terms $(SO)$}} + (1-\hat{\gamma}) \frac{(\mu_1 - \hat{\mu}_1)}{\hat{\mu}_1} - \frac{(\mu_0 - \hat{\mu}_0)}{\hat{\mu}_1} - ( \gamma - \hat{\gamma} ) \nonumber \\
& = SO + \frac{(\mu_1 - \hat{\mu}_1)}{\hat{\mu}_1} -\hat{\gamma}\frac{(\mu_1 - \hat{\mu}_1)}{\hat{\mu}_1} - \frac{(\mu_0 - \hat{\mu}_0)}{\hat{\mu}_1} - ( \gamma - \hat{\gamma} ) \nonumber \\
& \leq c \sum_{a\in \{ 0,1 \}} || \hat{\pi} - \pi || \cdot || \hat{\mu}_a - \mu_a || + c || \hat{\gamma} - \gamma || \cdot || \hat{\mu}_1 - \mu_1 ||. \label{eq:secondordererror}
\end{align*}
\endgroup
This final inequality follows from Cauchy--Schwartz $(P(f g) \leq ||f|| \cdot ||g||)$. Thus, we showed that Term III is a second-order error, meaning that it is a sum of a product of errors rather than a single error. This shows that if $\pi(x)$ and $\mu_a(x)$ are estimated nonparametrically, then for Term IV to be $o_\mathbbm{P}(n^{-1/2})$ it is required that $||\hat{\pi} - \pi||=o_\mathbbm{P}(n^{-1/4})$ \emph{and} $||\hat{\mu}_a - {\mu}_a||=o_\mathbbm{P}(n^{-1/4})$. 

Note that this result is slightly different than the analogous one for doubly robust (DR) estimators because we do not have a single product of the error term for the propensity score and the outcome regression. If one uses parametric models to estimate $\hat{\pi}(x)$ and $\hat{\mu}_a(x)$, to have an asymptotically normal result, the asymptotic normality (AN) of our estimator requires that the outcome regressions $\hat{\mu}_a(x)$ be consistent. In other words, as long as the outcome regression is properly specified, our estimator produces consistent results. This is unlike the DR estimators for which either $\pi(x)$ or $\mu_a(x)$ can be estimated properly to still arrive at a consistent estimator. However, we propose that these nuisance functions be estimated nonparametrically. With our estimator, there is no need to take the risk of having misspecified models since our estimator is AN with weak structural requirements.

\textbf{IV)} For Term IV, we assume that $g(x;\beta)$ is differentiable in $\beta$, so we can Taylor expand the first term, $P \varphi (\hat{\beta}),$ about $\beta$,
\begin{equation}
\label{eq:taylor}
P \varphi (\hat{\beta}) \approx P \varphi (\beta) + \frac{\partial P \varphi (\beta)}{\partial \beta} (\hat{\beta} - \beta) + \frac{1}{2} \frac{\partial ^2 \varphi (\beta)}{\partial \beta^2} ( \hat{\beta}-\beta)^2.
\end{equation}
The second-order term in the Taylor expansion is $o_\mathbbm{P}(n^{-1/2})$. This is because the second derivative is $O_p(1)$ since, by assumption, there exists a ball $B$ around $\beta$ such the second derivative is dominated by its norm for every $\beta$, and $(\hat{\beta} - \beta)^2$ is $o_\mathbbm{P}(n^{-1/2})$. So, when the two terms are multiplied together they are $o_\mathbbm{P}(n^{-1/2})$. For the first-order term, we let
\begin{equation}
M:= \frac{\partial P \varphi (\beta)}{\partial \beta} = P \frac{\partial \varphi (\beta)}{\partial \beta}.
\end{equation}
By passing the zeroth-order term to the left-hand side we get that Term IV is
\begin{equation}
\label{eq:betahatminusbeta}
P ( \varphi (\hat\beta) - \varphi (\beta) ) = M (\hat{\beta} - \beta) + o_\mathbbm{P}(n^{-1/2}),
\end{equation}
We are interested in a central-limit-theorem-type result, so we can isolate $(\hat{\beta} - \beta)$ to get
\begin{equation}
(\hat\beta - \beta) = (\mathbbm{P}_n - P) M^{-1} \varphi (\beta) +o_\mathbbm{P}(n^{-1/2}).
\end{equation}
It follows that the (multidimensional) asymptotic variance of $\hat{\beta}$ is
\begin{equation}
\label{eq:ansv}
\sqrt{n}(\hat\beta - \beta) \indist N \left(0, M^{-1} \mathbbm{E}(\varphi \varphi^T) (M^{-1})^T \right), 
\end{equation}

To summarize, the asymptotic results from the four terms are
\begin{equation}\small
0= \underbrace{(\mathbbm{P}_n-P) (\varphi (\hat{\beta}, \hat{\eta}) - \varphi (\beta, \eta))}_{o_\mathbbm{P}(n^{-1/2})} 
+ \underbrace{(\mathbbm{P}_n-P) \varphi (\beta, \eta)}_{\text{CLT}} 
+ \underbrace{P (\varphi (\hat{\beta}, \hat{\eta}) - \varphi (\beta, \hat{\eta}))}_{(\hat{\beta} - \beta) M + o_\mathbbm{P}(n^{-1/2})}
+ \underbrace{P (\varphi (\beta, \hat{\eta}) - \varphi (\beta, \eta) )}_{o_\mathbbm{P}(n^{-1/2})}. 
\end{equation}\normalsize
We combine them by using Slutsky's Theorem. Slutsky's Theorem says that, for two sequences of random variables $X_n$ and $Y_n$, that have $X_n \xrightarrow{p} X$ and $Y_n \xrightarrow{p} c$, for constant $c$, then $X_n + Y_n \xrightarrow{p} X + c$. Thus, by Slutsky's Theorem (applied twice since there are four terms), the sum of terms I, II, III, and IV converges to the normal distribution shown in \eqref{eq:ansv} (plus zero for the other terms), as desired.


\end{document}